\documentclass[11pt, a4paper]{article}
\usepackage[margin=2.54cm]{geometry}
\usepackage[utf8]{inputenc}
\usepackage[numbers]{natbib}
\usepackage{amsmath, amssymb, amsfonts, amsthm}
\usepackage{color}
\usepackage{mathtools}
\usepackage[hidelinks, pdfencoding=auto, psdextra]{hyperref}
\usepackage{cleveref}
\usepackage{jscommands}

\title{\vspace{-10mm}Rearrangement Events on Circular Genomes}
\author{Joshua Stevenson, Venta Terauds, Jeremy Sumner}
\date{\today}

\begin{document}

\pagenumbering{arabic}
\maketitle
 % end to do list

\begin{abstract}
    Early literature on genome rearrangement modelling views the problem of computing evolutionary distances as an inherently combinatorial one. In particular, attention is given to estimating distances using the minimum number of events required to transform one genome into another. In hindsight, this approach is analogous to early methods for inferring phylogenetic trees from DNA sequences such as maximum parsimony---both are motivated by the principle that the true distance minimises evolutionary change, and both are effective if this principle is a true reflection of reality. Recent literature considers genome rearrangement under statistical models, continuing this parallel with DNA-based methods, with the goal of using  model-based methods (for example maximum likelihood techniques) to compute distance estimates that incorporate the large number of rearrangement paths that can transform one genome into another. Crucially, this approach requires one to decide upon a set of feasible rearrangement events and, in this paper, we focus on characterising well-motivated models for signed, uni-chromosomal circular genomes, where the number of regions remains fixed. Since rearrangements are often mathematically described using permutations, we isolate the sets of permutations representing rearrangements that are biologically reasonable in this context, for example inversions and transpositions. We provide precise mathematical expressions for these rearrangements, and then describe them in terms of the set of cuts made in the genome when they are applied. We directly compare cuts to breakpoints, and use this concept to count the distinct rearrangement actions which apply a given number of cuts. Finally, we provide some examples of rearrangement models, and include a discussion of some questions that arise when defining plausible models.
\end{abstract}
   
\tableofcontents

\section{Introduction}
The goal of genome rearrangement modelling is to construct phylogenetic trees from a given set of genomes which share segments of DNA, but differ in the arrangement of these segments (a segment may be reversed, for example). These segments of DNA may be contiguous collections of genes (sometimes referred to as synteny blocks) or genes themselves. We will refer to these conserved segments of DNA simply as \textit{regions}. Genome rearrangements are events which affect segments of one or more regions of DNA within a genome: physically breaking the genome in one or more places and then reconnecting in some way. Often, phylogenetic trees are constructed by first estimating the evolutionary distance between elements of a given set of genomes, and then using a distance-based tree reconstruction method to obtain a phylogeny. Evolutionary distance estimates are obtained by considering the possible genome rearrangements which can be applied to transform a given genome into another.

The theory of genome rearrangements has developed in a similar way to a number of other areas in phylogenetics. For both the reconstruction of phylogenetic trees, and the more granular problem of computing evolutionary distance between genomes, the most obvious solution is to determine the minimum `amount' of evolution required, in accordance with the Occam's razor principle. In the case of tree reconstruction, this leads to `maximum parsimony' and related approaches---methods that the general phylogenetics literature has largely moved away from over the last few decades, due to problems with bias in certain situations, as was demonstrated by \citeauthor{Felsenstein1973} as early as \citeyear{Felsenstein1973} \cite{Felsenstein1973}. Similarly, the straight-forward combinatorial approaches to estimating evolutionary distance between genomes (namely, `minimal distance' under various sets of allowed rearrangements) have been successful in many ways, yielding fast algorithms for computing particular distances \cite{Hannenhalli1996, Bader2001}, but have had limitations in terms of flexibility of both the models and the very definition of rearrangement distance. 

It is becoming increasingly clear that the next phase of development of genome rearrangement theory will rely on model-based approaches  (for example maximum likelihood) with a flexible array of available distance measures and tree reconstruction methods (see \cite{Terauds2022, Terauds2018, EgriNagy2014, Bhatia2018}). In our view, developing effective model-based methods for computing genome rearrangement distances will require formalising existing concepts and specifically expanding upon them from an algebraic perspective. \citeauthor{Bhatia2018} \cite{Bhatia2018} make an excellent contribution toward this goal, clarifying a number of concepts used in the genome rearrangement literature. We aim to expand upon this contribution through the ideas and results we present here. With a focus on rearrangements, we connect numerous competing standards in the literature, centring on one particular representation of genomes that we argue is a natural choice for model-based approaches. Further, we provide an overview of the precise algebraic operations required to translate between different algebraic representations of genomes and rearrangements as detailed in \citepair{Bhatia2018}.

In this paper, we identify and precisely describe, in algebraic terms, the biologically significant classes of rearrangements---including the most prevalent cases which have appeared in the literature---and discuss these in the context of model-based methods. In \Cref{sec:gen_as_perms}, we give an overview of the various existing representations of genomes and rearrangements as permutations, and motivate our choice of representation, which we use for the remainder of the paper. \Cref{sec:ER} focuses on symmetry: in particular, we discuss how the symmetries of circular genomes and rearrangements inform the algebraic objects we use to represent rearrangement events, as well as the way in which we define models. We subsequently discuss rearrangement models from a practical point of view in \Cref{sec:models}, giving examples to illustrate some important considerations for choosing rearrangement probabilities and estimating genome rearrangement distances. We introduce and formally define the `set of cuts' applied by a given rearrangement in \Cref{sec:cuts}, and use this idea to characterise and for the first time establish explicit expressions for the number of distinct rearrangement actions which are precisely the biologically reasonable ones. We also discuss the relationship between this idea and the important and well-studied concept of a `breakpoint'. The paper concludes with a summary including a brief discussion of possible future work.

\section{Background}
\subsection{Genomes and Rearrangements as Permutations} \label{sec:gen_as_perms}
Approaches to genome rearrangement modelling can differ greatly depending on the particular types of genomes in question. We will focus on modelling uni-chromosomal circular genomes with a fixed number of oriented regions, and `genome' will refer to a genome under these assumptions unless otherwise stated. Genomes of this type are found in bacteria (with some exceptions); see, for example \cite{Thanbichler2006}. This choice excludes some rearrangement events (namely fissions, fusions, duplications and deletions) which change the number of chromosomes or the number of regions, however this still leaves our set of possible rearrangement events much larger than the set of all inversions, for example, which is a model commonly studied in the literature \cite{Hannenhalli1996, Caprara1997, Bafna1993, Berard2007, Bafna1996, Bader2001, Baudet2015, Lin2008}. 

The use of permutations to represent genomes is ubiquitous in the literature; however, even when given a particular set of constraints---such as the assumption that genomes are circular---there are numerous choices which need to be made in terms of the way we represent genomes and rearrangement events. \citepair{Bhatia2018} categorise these choices into two main paradigms, namely the `position' and `content' paradigms. The key difference is that in the position paradigm, the locations of regions in the genome are encoded via a fixed position labelling, whereas in the content paradigm, we only keep track of each region's location in relation to other regions. To complicate matters, even after choosing a paradigm, there remain multiple choices of mathematical representation. We presently explain one particular way to describe genomes using the position paradigm, and then explore other representations, including ones falling under the content paradigm.

To represent a genome in the position paradigm, we first label the $n$ different genomic regions with the set of integers $[n]:=\set{1,...,n}$. We then describe the position and orientation of each region as a map from the set of regions to the set of signed positions $[\ncol{n}] := \{1, ..., n, \ncol{1}, ... \ncol{n}\}$, where the bar above a position label indicates a reversed region. For example, consider a genome with 5 regions represented by the map $\sigma = \left[\begin{smallmatrix}
        1 & 2 & 3 & 4 & 5 \\
        \ncol{1} & 4 & 3 & \ncol{5} & \ncol{2}
\end{smallmatrix}\right]$, where the top row represents the pre-image (regions) and the bottom row represents the image (signed positions). Often the top row is omitted, which is unambiguous when the region labels are ordered as they are here. Note that in this form, the map $\sigma$ is not a bijection, because its pre-image is $[n]$. To remedy this, we extend the domain of $\sigma$ to all of $[\ncol{n}]$ via the rule
\begin{align}
        \sigma(\ncol{\,i\,}) := \ncol{\sigma(i)} , \label{eq:ext_perm}
\end{align}
for all $i\in [n]$.
The extended map $\sigma$ is a bijection and hence a permutation on the set $[\ncol{n}]$. The subset of permutations on $[\ncol{n}]$ which satisfy (\ref{eq:ext_perm}) forms a subgroup of $S_{[\ncol{n}]}$ (the group of permutations on $[\ncol{n}]$), namely the hyperoctahedral group, which we denote as $\HO_n$. This extension means we can write $\sigma$ simply as $\left[\begin{smallmatrix}
        1 & 2 & 3 & 4 & 5 \\
        \ncol{1} & 4 & 3 & \ncol{5} & \ncol{2}
\end{smallmatrix}\right]$, or in one-row notation as $[\ncol{1} 4 3 \ncol{5} \ncol{2}]$, remembering that its action on the remaining elements of $[\ncol{n}]$ is given by \Cref{eq:ext_perm}. We of course have a natural embedding of $\mathcal{S}_n$ as a subgroup of $\HO_n$, obtained by extending each permutation in $\Sy_n$ via the above process. In this paradigm, rearrangements are also expressed as elements of the hyperoctahedral group, and interpreted as maps from the set $[\ncol{n}]$ of signed positions to itself.

We now convert the genome $\sigma = \left[\begin{smallmatrix}
        1 & 2 & 3 & 4 & 5 \\
        \ncol{1} & 4 & 3 & \ncol{5} & \ncol{2}
\end{smallmatrix}\right]$ to illustrate some of the other commonly-used permutation representations of genomes. To aid our explanation, consider the two permutations $c=(1...n)(\ncol{n}...\ncol{1})$ and $r=(1\ncol{1})...(n\ncol{n})$ in $S_{[\ncol{n}]}$, written in cycle notation (i.e, $c$ maps $1\mapsto 2, 2\mapsto3, ...,  n\mapsto 1, \ncol{1}\mapsto\ncol{n}, \ncol{2}\mapsto\ncol{1}$ etc, and $r$ maps $1\leftrightarrow\ncol{1}$, ..., $n\leftrightarrow\ncol{n}$). Note that $c$ is not an element of $\HO_n$ whereas $r$ is. Using the permutations $c$ and $r$, we can convert $\sigma$ into three further representations of the same genome. These are detailed in \Cref{tbl:genome_representations}. 

\begin{table}[H]
    \begin{center}
    \begin{tabular}{ c c c c c }
     \textbf{expression} & \textbf{paradigm} & \textbf{version} & \textbf{sign means} & \textbf{example} \\
     $\sigma$ & position & region $\rightarrow$ position & orientation & $[\ncol{1}\ 4\ 3\ \ncol{5}\ \ncol{2}]$ \\ 
     $\sigma^{-1}$ & position & position $\rightarrow$ region & orientation & $[\ncol{1}\ \ncol{5}\ 3\ 2\ \ncol{4}]$ \\ 
     $\sigma^{-1}c\sigma$ & content & cycles & orientation & $(\ncol{1}\ncol{5}32\ncol{4})(4\ncol{2}\ncol{3}51)$\\ 
     $\sigma^{-1}c\sigma r$ & content & adjacencies & $\ncol{\text{head}}$/tail & $(1\ncol{5})(53)(\ncol{3}2)(\ncol{2}\ncol{4})(4\ncol{1})$
    \end{tabular}
    \end{center}
    \caption{Converting a permutation $\sigma$ representing an instance of a genome in the position paradigm (reg $\rightarrow$ pos) to various other paradigms. Recall that since $n=5$ in this example, we have $c=(12345)(\ncol{5}\ncol{4}\ncol{3}\ncol{2}\ncol{1})$ and $r=(1\ncol{1})(2\ncol{2})(3\ncol{3})(4\ncol{4})(5\ncol{5})$.}
    \label{tbl:genome_representations}
\end{table}

Note that each representation is defined based on the interpretation of the input and output of the permutation. The position paradigm representations are maps from regions to positions or from positions to regions. The content paradigm representations are maps from each region to the following region (in a clockwise direction) or maps between heads and tails of regions (in this case, the sign of the region number indicates either a head or tail, rather than orientation of a region). Note that the examples in the final column of \Cref{tbl:genome_representations} are written in a notation chosen to most easily express the genome in the particular paradigm. It is also common to see permutations representing position paradigm genomes, for example, written in cycle notation as well.

The expressions in the first column in \Cref{tbl:genome_representations} are concise because we started in the position paradigm (which in one sense contains the `most' data). If we began with a genome in the content paradigm, obtaining the equivalent genome in the position paradigm is a less pleasant process. For example, to transform $\pi = (\ncol{1}\ncol{5}32\ncol{4})(4\ncol{2}\ncol{3}51)$ (content/cycles paradigm) in \Cref{tbl:genome_representations} into $\sigma^{-1} = [\ncol{1}\ \ncol{5}\ 3\ 2\ \ncol{4}]$ (position paradigm, positions-to-regions), we must consider one of the cycles as the bottom row of a signed permutation. Of course, there are two cycles to choose from and each may be written in $n$ different ways, due to dihedral symmetry of the genome. If we decide to choose the cycle containing `1', we can write $\sigma^{-1}$ down in a more systematic way, by setting $\sigma^{-1}(1) = 1$,
$\sigma^{-1}(i) = \pi^{i-1}(1)$ for $i\neq 1$ and then extending via $\sigma^{-1}(\ncol{i}) = \ncol{\sigma^{-1}(i)}$. Still, this is far from being as nice as the simple algebraic expressions in the first column of \Cref{tbl:genome_representations}.

Paradigm choice also determines how rearrangements are represented and applied. For example, in the position paradigm where genomes are maps from regions to positions, rearrangements are thought of as maps from positions to positions. In this case, applying these rearrangements is done via composition. For example, if $\sigma$ represents a genome and $\alpha$ represents a rearrangement, then $\alpha \sigma$ represents a genome where the positions of the regions have been shuffled around by $\alpha$; region $i$ is now in position $\alpha(\sigma(i))$. In the content paradigm, considering genomes to be products of two $n$-cycles (as in the third row of \Cref{tbl:genome_representations}), rearrangements are maps from regions to regions, and are applied via conjugation. If $\sigma^{-1}c\sigma$ represents a genome in the content paradigm, then $\sigma^{-1}c\sigma(i)=j$ means that region $i$ is followed by region $j$. For example, the rearrangement $\alpha = (jk)$ maps region $j$ to region $k$ and is applied via conjugation, and thus $\alpha^{-1}\sigma^{-1}c\sigma\alpha(i)$ will map $i$ to $k$.

A clear benefit of the content paradigm is that we have a 1-1 correspondence with circular genomes, without needing to resort to cosets or formal sums of permutations, as is the case when working in the position paradigm (see \Cref{sec:ER} for details). An important disadvantage though, is that the outcome of `applying' a rearrangement to a genome in the content paradigm depends on the genome itself. For example, \citeauthor{Meidanis2000} state: \textit{``There is no way to define a class of permutations that will be `the reversals', valid for all genomes''} \cite{Meidanis2000}, referring to defining inversions as permutations to be applied by left multiplication in the content paradigm. This makes it difficult to algebraically define a model in which not all rearrangements are equally probable, although an appropriate group action can be defined by converting to the position paradigm, applying the rearrangement and then converting back. For example, the genome instances are written $\sigma c \sigma^{-1}$, and a rearrangement instance $\alpha$ is applied using the group action defined via $\sigma c \sigma^{-1} \longmapsto \sigma \alpha c \alpha^{-1} \sigma^{-1}$, where $\alpha$ is a position paradigm rearrangement instance and $\sigma$ is a position paradigm genome instance. Under this group action, we can in fact write down a set of all permutations which represent reversals, for example.

Conversely, in the position paradigm, we can (in theory) define a set of rearrangements which performs a specific type of action (for example, inversions of two adjacent regions). In \Cref{sec:cuts}, we define `$k$-cut' rearrangements in the position paradigm. An important difference between our definition and the definition of a `$k$-break' rearrangement in the content paradigm (algebraically, conjugation of a genome by a $k$-cycle \cite{Feijao2013, Bhatia2018}), is that the action of a $k$-cut rearrangement will always perform exactly $k$ cuts, whereas, the action of a $k$-break operation may differ depending on the genome it is being applied to. Even when the rearrangement does perform $k$ cuts, the distribution of cuts across the genome will depend on the genome itself. While this disadvantage makes representing genomes and rearrangements using the content paradigm for model-based approaches impractical, the content paradigm has provided the field with a number of important results and efficient algorithms. As described in \citepair{Bhatia2018}, moving between paradigms can provide additional understanding. In fact, in \Cref{sec:cuts} follows we will use the permutation $c=(1...n)(\ncol{n}...\ncol{1})$ to aid our description of different classes of rearrangements.

\subsection{Circular symmetries and equivalent rearrangements} \label{sec:ER}

\begin{figure}
    \centering
    \includegraphics[width=0.9\textwidth]{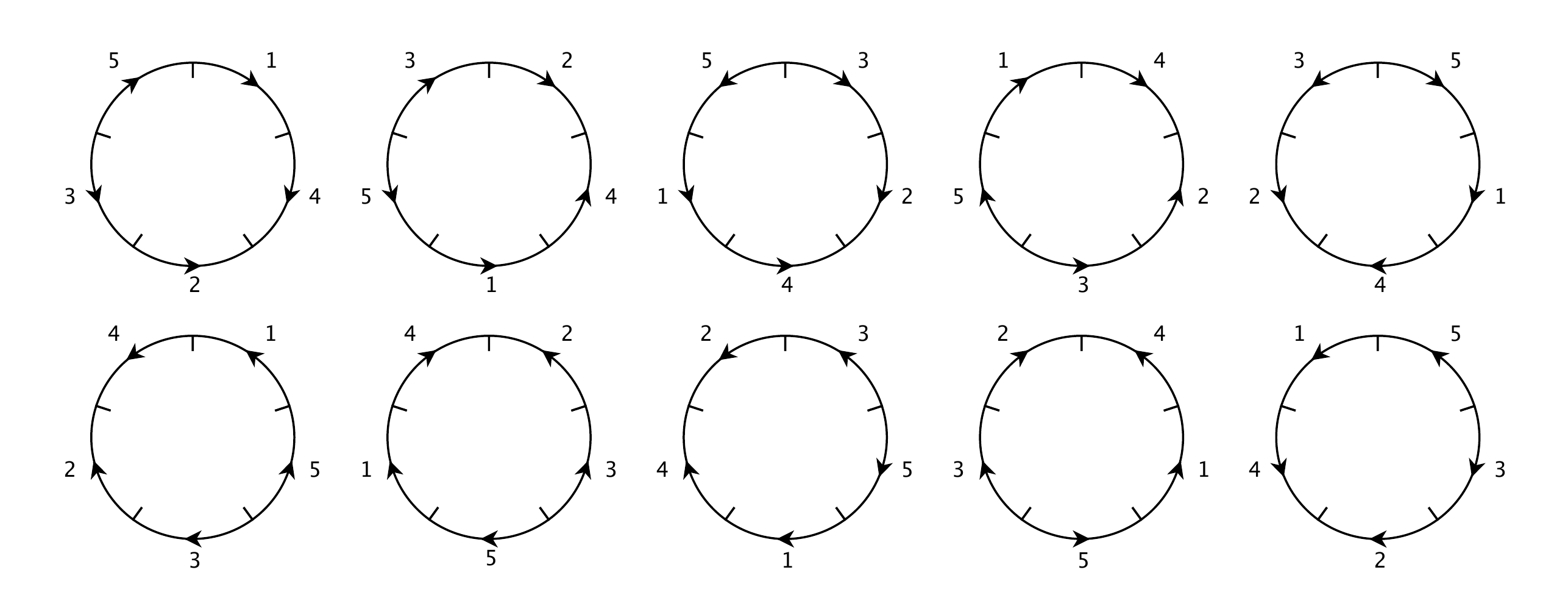}
    \caption{Diagrams representing each of the $2n=10$ instances of the example genome $\mathcal{D}_n\sigma$ arising from its dihedral symmetries. The bottom left diagram matches the instance $\sigma = \left[\begin{matrix}
            \ncol{1} & 4 & 3 & \ncol{5} & \ncol{2}
        \end{matrix}\right]$ where positions run clockwise starting from position 1 in the top-right of the circle.}
    \label{fig:example}
\end{figure}

In \Cref{sec:gen_as_perms}, we discussed the fact that in the content paradigm, a single $k$-break operation can represent a number of physically different `rearrangements' depending on the genome it is applied to. In the position paradigm, the converse can occur for both genomes and rearrangements, but can be easily corrected as we will see shortly. In the following, we consider circular genomes under the position paradigm (regions $\rightarrow$ positions). Whether regions are oriented or otherwise, circular genomes modelled with no distinguished positions have dihedral symmetry---that is circular genomes are physical objects that are unchanged by rotations and flips in physical space. Because of this, different signed permutations can correspond to the same genome. For example, \Cref{fig:example} shows several diagrams representing a particular genome, each portraying the same genome in a different physical orientation, and thus corresponding to a different signed permutation. To express this algebraically, we consider which rearrangements in $\HO_n$ leave the \textit{genome} unchanged.
In our formulation, a single position rotation of a genome is given by the rearrangement \[r := (123...n)(\ncoverline{1}\ncoverline{2}\ncoverline{3}...\ncoverline{\vphantom{1}n}) \in \HO_n.\]
This is simply the analogous rotation in the dihedral group $\mathcal{D}_n$ extended to $\HO_n$. Reflections are more complicated, however, as we will see. For $n=5$, an example of a reflection is \[f := (2,\ncoverline{5})(\ncoverline{2},5)(3,\ncoverline{4})(\ncoverline{3},4)(1,\ncoverline{1}) \in \HO_n.\]
Of course, we cannot obtain this by taking the permutation $(2,5)(3,4) \in \mathcal{D}_5$ which represents a reflection for a genome with unoriented regions, and simply extending to $\HO_5$ via \Cref{eq:ext_perm}, as one might first think. This is because for genomes with oriented regions, reflections change the orientation of the regions. Regardless, the group generated by these two signed permutations is precisely the set of rearrangements which leave a given genome unchanged, and is isomorphic to $\mathcal{D}_n$. This can be seen by considering the definition of the dihedral group involving two generators, $\left\langle{r}, {f} \mid {r}^{n}={f}^{2}=({fr})^{2}=e\right\rangle \cong \mathcal{D}_n$. Throughout the paper we will denote this copy of the dihedral group as $\D_n$. 

Dihedral symmetry can be incorporated into calculations by considering genomes as cosets rather than individual permutations, as is already outlined in the literature \cite{Sumner2017, Terauds2018, Terauds2022}. In this case, each genome corresponds to a particular coset. For example, the genome corresponding to the permutation $\sigma$ is given by the coset $\D_5\sigma = \{ d\sigma : d \in \D_5 \}$, collecting all the reflections and rotations of $\sigma$, which we individually refer to as \textit{instances} of the genome $\D_5\sigma$. While cosets are a natural way to factor out symmetries, we will see shortly that it is more appropriate to use a different set of mathematical objects to represent genomes and rearrangements.

Genome rearrangements can of course also exhibit symmetry. When we multiply a permutation ($\alpha: \mathrm{positions} \rightarrow \mathrm{positions}$) representing a rearrangement into an instance of a genome $d\sigma \in \mathcal{D}_n\sigma$, the instance $\alpha d \sigma$ that we get back might not represent the same genome as the result of $\alpha$ acting on some other instance, say $d_2\sigma \in \mathcal{D}_n\sigma$. We cannot choose another element of $\D_n \alpha$ which achieves this. For example, the permutations $(1\ncol{1})$ and $(1\ncol{1})(12345)(\ncol{12345})$, are both expressions showing the application of a rearrangement $(1\ncol{1})$ to an instance of the the reference genome (i.e, an element of $\D_n$), but the resulting permutations are not both instances of the same genome. For this reason, \citeauthor{Terauds2018} \cite{Terauds2018} consider formal sums of rearrangements corresponding to \textit{double cosets} $\D_n \alpha \D_n$, $\alpha \in \HO_n$. An element of one of these double cosets looks like $d_1\alpha d_2$, with $d_1,d_2 \in \D_n$. Now, we can choose $\alpha d(d_2)^{-1}$ in the double coset, such that $d_2 \sigma$ is mapped to the same genome as $d\sigma$. Using formal sums instead of double cosets allows one to keep track of the probabilities of the possible outcomes when a rearrangement is applied to a genome. To allow ourselves to take sums of group elements, we move from the group itself to the group algebra $\mathbb{C}[\HO_n]$---a vector space (over the complex numbers, although in practice our coefficients will always lie between 0 and 1) which has the group as its basis, where we also allow ourselves to take products of elements. The product is simply the same as in the group, and is distrubuted over addition. As an example, we will describe the genome represented by $\sigma = [\ncol{1} 4 3 \ncol{5} \ncol{2}]$, used as an example in \Cref{fig:example} as an element of the group algebra.

Define the \textit{symmetry element}, \[ \mathbf{z} := \tfrac{1}{2n} \sum_{d \in \mathcal{D}_n} d \ \in\  \mathbb{C}[\HO_n] \] where the factor of $\frac{1}{2n}$ is so that the coefficients of the permutations in the sum all add to 1. For example, take the instance $\sigma = [\ncol{1} 4 3 \ncol{5} \ncol{2}]$. The product,
\begin{align*}
\textbf{z}\sigma = \sum_{d \in \mathcal{D}_n} d\sigma = \frac{1}{10} \Big(&[1, \ncol{3}, \ncol{4}, 2, 5] + [2, \ncol{4}, \ncol{5}, 3, 1]+
 [3, \ncol{5}, \ncol{1}, 4, 2]+
 [4, \ncol{1}, \ncol{2}, 5, 3]\\[-0.8em]
 +\  & [5, \ncol{2}, \ncol{3}, 1, 4]+
[\ncol{1}, 4, 3, \ncol{5}, \ncol{2}]+
 [\ncol{2}, 5, 4, \ncol{1}, \ncol{3}]+
 [\ncol{3}, 1, 5, \ncol{2}, \ncol{4}]\\
 +\  & [\ncol{4}, 2, 1, \ncol{3}, \ncol{5}]+
 [\ncol{5}, 3, 2, \ncol{4}, \ncol{1}]\Big),
\end{align*}
encodes every instance of the genome represented by $\sigma$.

The \textit{genome algebra}, defined by \citepair{Terauds2018}, is the sub-algebra of $\mathbb{C}[\HO_n]$ spanned by the set of genomes $\{\textbf{z}\sigma : \sigma \in \HO_n\}$ (we can think of the basis as corresponding to the set of cosets $\HO_n/\D_n$). Similarly, the analogue of the double coset $\D_n \alpha \D_n$ is the algebra element $\textbf{z}\alpha\textbf{z}$. Formally, a rearrangement $\textbf{z}\alpha$ acts on a genome $\textbf{z}\sigma$ via multiplication, producing a weighted formal sum $(\textbf{z}\alpha) (\textbf{z}\sigma) = (\textbf{z}\alpha\textbf{z})\sigma$\  of genomes, which represents the possible outcomes of the rearrangement event, along with the probability of each outcome occurring. It is also helpful to view this expression as,
\[(\textbf{z}\alpha) (\textbf{z}\sigma) = \textbf{z}(\alpha(\textbf{z}\sigma)) = \text{all orientations of ($\alpha$ acting on (all orientations of the instance $\sigma$)))}.\]
Two rearrangements $\textbf{z}\alpha$ and $\textbf{z}\beta$ applied to the genome $\textbf{z}$ will produce the same distribution of outcomes (that is, they apply the same action) precisely when $\textbf{z}\alpha\textbf{z} = \textbf{z}\beta\textbf{z}$, or equivalently when $\beta \in \D_n\alpha\D_n$. This set is the set of permutations generated by left multiplication by $\D_n$ (representing flips and rotations of the resulting genome) and conjugation by $\D_n$ (allowing the rearrangement to occur at any position on the genome, since we have no privileged positions). We refer to a set of allowed rearrangement actions $\mathcal{M} = \set{\textbf{z}\alpha_1 \textbf{z}, ..., \textbf{z}\alpha_m \textbf{z}}$ along with a probability map $w : \mathcal{M} \rightarrow [0,1]$ as a \textit{model} of genome rearrangement.

\begin{remark}
We note that in a number of situations, for example our combinatorial analysis of rearrangements in \Cref{sec:cuts}, it is sufficient to simply think of rearrangements as double cosets, or more simply, sets of permutations which represent a given rearrangement, along with some number of flips or rotations applied before or after the rearrangement itself. While imagining genomes and rearrangements as formal sums belonging to some algebra (such as the genome algebra) allows us to easily keep track of model probabilities and the outcomes of rearrangement actions, as we will see in \Cref{sec:models}, fundamentally we are simply collecting all instances of the genome or rearrangement which we consider to be equivalent up to symmetry. 
\end{remark}

\section{Rearrangement models} \label{sec:models}
\begin{table}
    \centering
    \renewcommand{\arraystretch}{1.5}
    \begin{tabular}{P{4cm}P{3cm}P{7cm}}
         Model & References & Comments \\ \hline 
         \textit{All inversions} & \cite{Hannenhalli1996, Caprara1997, Bafna1993, Berard2007, Bafna1996, Bader2001, Baudet2015, Lin2008} & Used for `reversal distance' or `inversion distance'. When restricted to single-chromosome circular genomes, these are also the rearrangements arising from the `double-cut-and-join' operator (DCJ) \\
         \textit{Inversions of one or two regions} & \cite{Oliveira2019, Galvao2017} & ``Super short reversals" \\
         \textit{Inversions of one, two or three regions} & \cite{Terauds2018}* &  \\ 
         \textit{Inversions or transpositions of one or two regions} & \cite{Oliveira2018, Oliveira2019} & ``Super short operations" \\
         \textit{Transpositions of one or two regions} & \cite{Oliveira2019} & ``Super-short transpositions" \\
         \textit{All transpositions and inversions} & \cite{Bafna1998, Walter2005} & While transposition distance is less well-studied than inversion distance, a number of  \\
         \textit{Block-interchanges} &  \cite{Huang2010} & Swap two non-intersecting sections of some number of regions.\\
    \end{tabular}
    \caption{Some examples of classes of allowed rearrangements appearing in the literature. *Since only unsigned genomes were discussed in this work, single-region inversions leave the genome unchanged.}
    \label{tbl:models_in_literature}
\end{table}

It is widely accepted by biologists that some genomic rearrangement events are more likely than others \cite{Dalevi2002, Lefebvre2003, Darling2008, Alexeev2015} (although there has been much debate over where the most likely locations are for rearrangements to occur \cite{Peng2006, Pevzner2003}). Some rearrangement event modelling approaches allow for only a specific subset of rearrangements, while others are valid for any subset. One of the most commonly seen sets of rearrangements considered in the literature is the set of all inversions \cite{Bader2001, Watterson1982, Hannenhalli1996}. Some other commonly used sets of rearrangements are outlined in \Cref{tbl:models_in_literature}. In the position paradigm, we can give explicit algebraic characterisations for the instances of some specific types of rearrangements, as follows.

\begin{proposition} \label{prop:inversion}
    Fixing a position labelling for a genome with $n$ oriented regions, inversions of odd length $2k+1$ about a signed region in position $r+1$ can be written as the product of transpositions
    \[(r\!+\!1, \ncoverline{r\!+\!1\vphantom{)}})\prod_{i=1}^k \big(r\!-\!(i\!-\!1), \ncoverline{r\!+\!(i\!+\!1)}\big) \big(\ncoverline{r\!-\!(i\!-\!1)}, r\!+\!(i\!+\!1)\big) \in \HO_n,\]
    where all addition is $(\mathrm{mod}\ n)$. Inversions of even length $2k$, about a pair of regions in positions $(r,r+1)$ can be written
    \[\prod_{i=1}^k \big(r\!-\!(i\!-\!1), \ncoverline{r\!+\!i\vphantom{)}}\big)\big(\ncoverline{r\!-\!(i\!-\!1)}, r\!+\!i\big) \in \HO_n.\vspace{-8mm}\]
    \hfill\qed
\end{proposition}

Of course, the set of all 2-inversions generates the symmetric group, and the collection of 1- and 2-inversions generates $\HO_n$. This of course means that any transposition, for example, can be written as a sequence of inversions. Once transpositions are added to the model as independent rearrangements, however, there is more to think about in terms of their assigned probability and further rearrangements which may need to be added as a result. For example, does it make sense to include a transposition, but not include the same transposition that also inverts the segment(s) involved? We will now discuss questions such as this one and provide some examples. In each example, we describe a set of possible rearrangements, along with a relative probability. This information can be combined to form a Markov chain in which the states are the genomes themselves, or used with some other method for computing rearrangement distance.

\subsection{Examples} \label{sec:examples}
Identifying biologically realistic models for genome rearrangements along with associated probabilities is still a work in progress. It is therefore important for any proposed framework to allow for a range of different models. The following examples will show some commonly studied models (for example `all inversions') in the genome algebra framework. We will discuss some potential pitfalls and important considerations which arise when defining a model and choosing rearrangement probabilities. 

\begin{example}(All inversions, all equally likely).
    Suppose we wish to construct a model equivalent to that used to compute `inversion distance' on circular genomes; that is, a model containing every inversion, with each inversion having an equal probability of occurring. See \Cref{fig:grah} for a visual representation of such a model.
    
    Begin by including the inversion represented by  $(1\ncol{1})$ (an inversion of the single region in position 1). The action of the rearrangement is the same as the action of $(2\ncol{2})$, for example, since the genome can be rotated or flipped before and after applying the rearrangement due to dihedral symmetry. As previously discussed, the sum of permutations $\zed (1\ncol{1}) \zed$ represents the single region inversion action, and thus accounts for all such inversions. Similarly, $\zed(1\ncol{2})(\ncol{1}2)\zed$ accounts for every inversion of 2 adjacent regions, and  $\zed(1\ncol{3})(2\ncol{2})(\ncol{1}3)\zed$ gives us all of the 3-region inversions. Care must be taken for larger inversions, however, since for a circular genome with $n$ regions, the $(n\!-\!m)$-region inversion action is the same as the $m$-region inversion action. If $n=6$, for example, the three inversions we have mentioned so far account for all inversions. If $n=5$, for equal probabilities we must remove the 3 region inversion $\zed(1\ncol{3})(2\ncol{2})(\ncol{1}3)\zed$ as it is equivalent to an inversion of two regions and would therefore skew the probability distribution of outcomes if included. The effect that this can have on subsequent calculations of genomic distance is illustrated in \citepair{Terauds2021}.
    
    It is important to note that the model $\mathcal{M} = \set{\zed (1\ncol{1}) \zed, \zed(1\ncol{2})(\ncol{1}2)\zed, \zed(1\ncol{3})(2\ncol{2})(\ncol{1}3)\zed}$ with an equal weighting of $\frac{1}{3}$ for each action will produce (via the process detailed by \citepair{Terauds2018}) a Markov chain in which some outcomes reachable via a single step will be more likely than others. This is because the action of inverting of $3$ regions has 3 possible outcomes, whereas each of the other inversion actions have 6. In other words, the sum $\zed(1\ncol{3})(2\ncol{2})(\ncol{1}3)\zed$ has fewer terms. The model for which \textit{all reachable genomes are equally likely} can be achieved via assigning the probability $\frac{2}{5}$ to the 1- and 2-region inversions, and $\frac{1}{5}$ to the 3-region inversion.
    
    Note that in total, there are $\floor{\frac{n}{2}}$ distinct inversion actions for a genome with $n$ regions. This is shown by \Cref{prop:num_inversions}.
\end{example}

\begin{example}(One and two region inversions).
A model containing small inversions of one or two regions can be defined with the two sums of permutations, $\zed (1\ncol{1}) \zed$ and $\zed(1\ncol{2})(\ncol{1}2)\zed$, as above. Since the number of regions $n$ is encoded in the symmetry element $\zed$, this model can be written in the same way for any number of regions $n > 2$. These rearrangement actions can be allowed to occur with equal probability by simply assigning $\frac{1}{2}$ to each action, or they can be given different weights. For example, there is some evidence of smaller inversions being more prevalent compared to larger inversions \cite{Dalevi2002, Lefebvre2003}, and this can be modelled by assigning the 1- and 2-region inversions the probabilities $\frac{2}{3}$ and $\frac{1}{3}$, for example. In this case, the 1-region inversion is twice as likely to occur.
\begin{figure}
    \centering
    \includegraphics[width=0.6\linewidth]{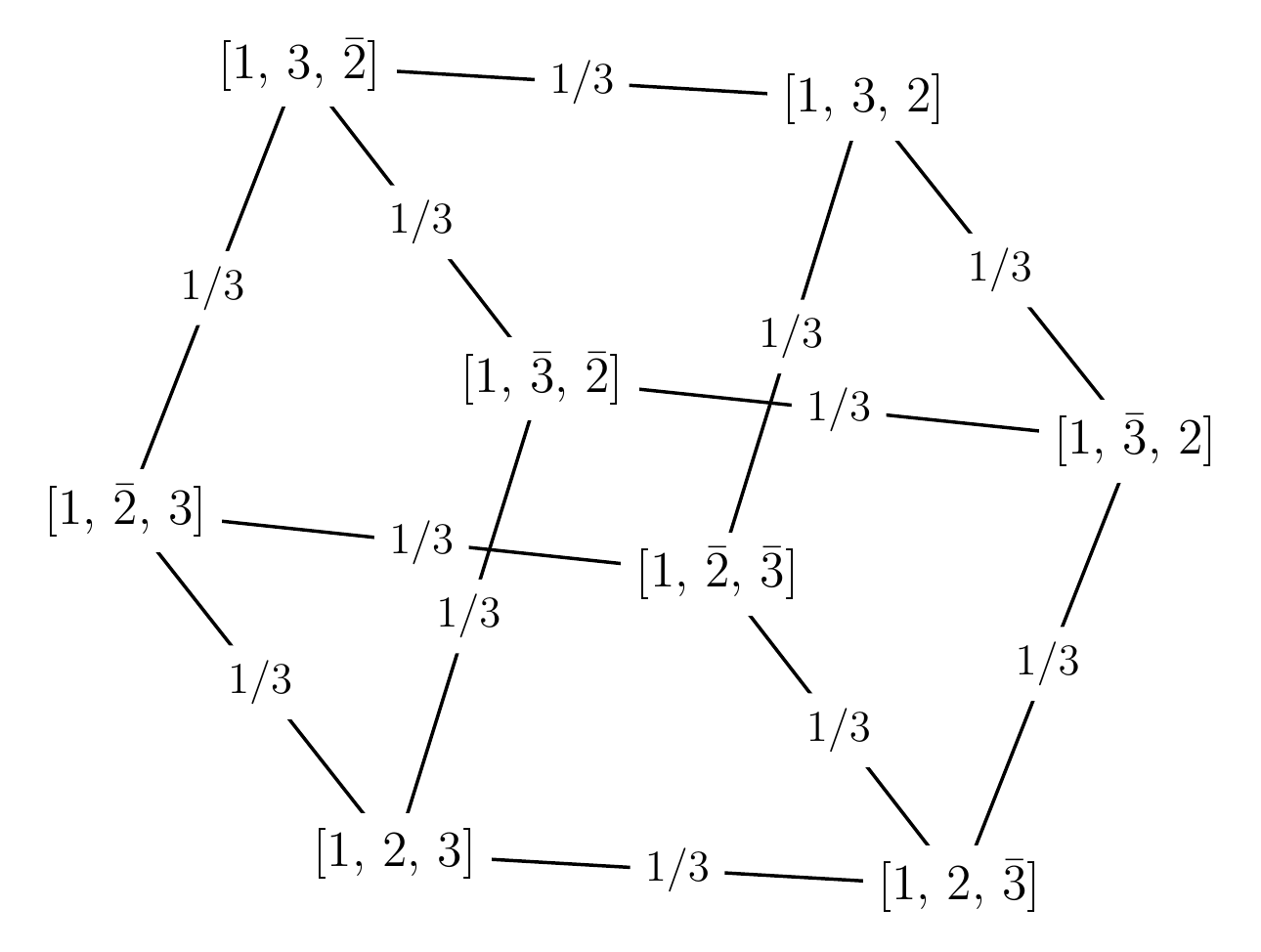}
    \caption{A graph representation of the Markov chain for a model of inversions on circular genomes with 3 oriented regions. Lines between genomes show the probability of a rearrangement which transforms one into the other. The minimum number of steps needed to traverse the graph from one genome to another is the `inversion distance'. The edge weights are the same for any choice of probability mapping $w$, since every inversion of two regions can be written as an inversion of a single region.}
    \label{fig:my_label}
\end{figure}
\end{example}

\begin{example}(One and two region inversions, and two region adjacent transpositions).
Suppose we have a model with small (1- and 2-region) inversions as in the previous examples. We can add adjacent transpositions (rearrangements that swap two regions) by including the additional element $\zed (12)(\ncol{1}\ncol{2}) \zed$. In this situation, it might make sense to also allow inverting one or both of these swapped regions in a single step. The action $\zed (12\ncol{1}\ncol{2}) \zed$ swaps regions in adjacent positions, while simultaneously inverting the second region (the action that inverts the first region is equivalent due to symmetry). In this case, one should consider a weighting for this new action that makes sense considering the other weightings. 

For example, if $w(\zed (1\ncol{1}) \zed) = a$ and $w(\zed (12)(\ncol{1}\ncol{2}) \zed) = b$, then it might make sense to require that $w(\zed (12\ncol{1}\ncol{2}) \zed) \geq ab$. One can think of this property as the triangle inequality for a rearrangement model.
\end{example}

\begin{example}
Another commonly considered class of rearrangement events is the collection of \textit{transpositions}. These are rearrangements which cut a segment from the genome and move it to a new location and/or orientation. In the previous example, $\zed (12\ncol{1}\ncol{2}) \zed$ and $\zed (12)(\ncol{1}\ncol{2})\zed$ were examples of transpositions of adjacent regions. Another permutation representing a transposition is $(132)(\ncol{1}\ncol{3}\ncol{2})$, which takes the region in position 1 and moves it to position 3. More generally, the action $\textbf{z}(132)(\ncol{1}\ncol{3}\ncol{2})\textbf{z}$ takes a region and moves it to the location two regions further along the genome. Of course, there are many more transposition actions. For $n=5$, we have 6 such actions: two correspond to moving a single region 1 or 2 places along the genome, and for each of these, we can invert the region being moved, or invert one of the other segments, or both the single region and one of the segments. Of course, when there are only $5$ regions, moving a segment of two regions along the genome is equivalent to moving a single region. If inversions are included in the model as single-step rearrangements, one might want to remove transpositions which invert one of the segments. This can be done by simply removing rearrangements with cycles that contain positive and negative region labels (for example $\zed(12\ncol{1}\ncol{2})\zed$). We will examine these rearrangements more closely in \Cref{sec:cuts}, where we count the distinct rearrangements of this type.
\end{example}

\section{Cuts and breakpoints} \label{sec:cuts}
While the examples in the previous section provide some intuition for defining genome rearrangement models, it can be useful to formally describe possible rearrangements in order to develop a clearer picture. In this section, we consider the categorisation of rearrangements via the positions of cuts they make in the genome, and find expressions for the numbers of distinct rearrangements occurring via two and three cuts. For example, given a position labelling, the permutation $(2\ncol{3})(\ncol{2}3)$ represents an instance of an inversion which cuts the genome in two places. That is, between positions 1 and 2 and positions 3 and 4. Formally, we introduce the concept of the set of cuts applied by a permutation.

\begin{definition}  \label{def:cuts}
    For $\alpha \in \HO_n$, the \textit{set of cuts applied by $\alpha$} is given by, \vspace{-2mm}
    \[ \mathrm{cuts}(\alpha) := \set{\ i\in [n]\ : \text{ the action of $\alpha$ cuts a genome between positions $i$ and $i\!+\!1 \ (\text{mod}\ n)$ }}.\vspace{-0.7cm}\] 
\end{definition}

This idea is similar to the important (although somewhat overloaded, with widely differing definitions) term, \textit{breakpoint}. In short, the breakpoints are the locations on the genome at which rearrangements can or have occurred. We will discuss breakpoints and the similarities to \Cref{def:cuts} later in this section.

Recognising that for a rearrangement instance $\alpha$, the set of cuts is the set of positions which $\alpha$ places `out-of-order' when applied to the identity genome, we have the following,

%\begin{samepage}
\begin{proposition}  \label{prop:cut_set}
    For $\alpha \in \HO_n$, the \textit{set of cuts} applied by a rearrangement instance $\alpha$ is given by,
    \[ \mathrm{cuts}(\alpha) := [n] - \mathrm{Fix}([\alpha,c]) \]
    where $[\alpha,c] := \alpha^{-1} c^{-1}\alpha c \in \HO_n$ is \textit{the commutator of} $\alpha$ \textit{and} $c$, and \\$\mathrm{Fix}([\alpha,c])\!=\!\set{i\!\in\![n] : [\alpha,c](i)\!=\!i}$, recalling that $c=(1,...,n)(\ncol{n},...,\ncol{1}) \in \Sy_{[\ncol{n}]}$.
\end{proposition}
\begin{proof}
    Consider $\alpha \in \HO_n$. We have $i \in \cuts(\alpha)$ if and only if the region following $\alpha(i)$ is different from the region $\alpha(i\!+\!1\ (\mathrm{mod}\ n)) = \alpha c (i)$. That is, the regions in positions $i$ and $j=c(i)$ are out of order with respect to the identity.
    We note that if $i \in \cuts(\alpha)$ then $(i,j)$ is a pair of consecutive labels in $\alpha$. The permutation $c$ tells us the next position on the circular genome. Formally, we can write,
    \begin{align*}
        \cuts(\alpha) &= \{ i \in [n] : c\alpha(i) \neq \alpha c (i) \} \\
        &= \set{i \in [n] : i \neq \alpha^{-1}c^{-1}\alpha c(i)} \\
        &= \left\{ i \in [n] : i\notin\mathrm{Fix}([\alpha,c])\right\} \\
        &= [n] - \mathrm{Fix}([\alpha,c])
    \end{align*}
    as required.
\end{proof}

\begin{example}
    As in \Cref{prop:inversion}, the permutation $\alpha = (3\ncol{3})(2\ncol{4})(\ncol{2}4)$ represents a rearrangement that cuts a genome in two positions to invert the 3-region segment spanning positions 2, 3 and 4. We have
    \begin{align*}  
        [\alpha, c] &= \alpha^{-1}c^{-1}\alpha c \\
        & = (3\ncol{3})(2\ncol{4}) (\ncol{2}4) (\ncol{1}\ncol{2}\ncol{3}\ncol{4}\ncol{5})(54321)(3\ncol{3})(2\ncol{4}) (\ncol{2}4) (12345)(\ncol{5}\ncol{4}\ncol{3}\ncol{2}\ncol{1}) \\
        & = (1\ncol{5})(4\ncol{2}) 
    \end{align*}
    which fixes $2, 3, 5, \ncol{1}, \ncol{3}$ and $\ncol{4}$. Thus $\cuts(\alpha) = [n] - \mathrm{Fix}([\alpha,c]) = \{1,4\}$, as expected since $\alpha$ cuts the genome before position 1 and after position 4.
\end{example}

Note that the expression given in \Cref{prop:cuts_set} is very similar to the expression for the number of breakpoints derived by \citeauthor{Meidanis2000} \cite{Meidanis2000}. In fact, breakpoints are a similar concept, but are typically sets of regions rather than sets of positions. We will obtain an equivalent expression for the number of breakpoints in \Cref{sec:breakpoints}, using a similar construction to the cuts set.

As discussed in \Cref{sec:ER}, in our formulation which incorporates circular symmetry, rearrangements aren't defined to be single permutations, but rather cosets or formal sums. The following result shows that all instances of a given rearrangement produce the same cut set, and thus the cut set is well-defined on the set of rearrangements. We first need the following definition.

\begin{definition} \label{def:ccuts}
     For all permutations $\alpha \in \HO_n$, denote \[\mathrm{Cuts}(\alpha) := \cuts(\alpha) \cup \ncol{\cuts(\alpha)},\]
     where $\ \ncol{\cuts(\alpha)} := \set{\ncol{\,i\,} : i \in \cuts(\alpha)}$.
\end{definition}

Simply put, \Cref{def:ccuts} describes a set which includes \textit{both} orientations of the positions in the cut set $\cuts(\alpha)$. If we take $\alpha$ from Example 5.3 above, we get $\mathrm{Cuts}(\alpha) = \{1,4,\ncol{1},\ncol{4}\}$. This idea allows us to more concisely state the following result.

\begin{proposition} \label{prop:cuts_set}
    For all permutations $\alpha \in \HO_n$ and $d_1,d_2\in \mathcal{D}_n$, we have,
    \[ \mathrm{Cuts}(d_1\alpha d_2) = \{ d_2^{-1}(i) : i \in \mathrm{Cuts}(\alpha) \} =:  d_2^{-1}\cdot \mathrm{Cuts}(\alpha),\]
    Further,
    \[ \cuts(d_1\alpha d_2) = \paren{d_2^{-1}\cdot \mathrm{Cuts}(\alpha)} \cap [n],\]
    That is, if two rearrangements have the same action, then their cut sets are dihedrally related. In particular,
    \[ \cuts(d_1\alpha) = \cuts(\alpha).\]
\end{proposition}
\begin{proof}
    Take some $\alpha \in \HO_n$ and $d_1,d_2\in \mathcal{D}_n$. Then,
    \begin{align*}
        \mathrm{Cuts}(d_1\alpha d_2)
        &= \left\{ i\in [\ncol{n}] : \ cd_1\alpha d_2(i) \neq d_1\alpha d_2c(i)\right\} \\
        &= \left\{ i\in  [\ncol{n}] : \ d_1^{-1}cd_1\alpha d_2(i) \neq \alpha d_2c(i)\right\} \\
        &= \left\{ i\in  [\ncol{n}] : \ c\alpha d_2(i) \neq \alpha c d_2(i)\right\} \tag*{\hspace{-1cm}$\because$\ \ $c$ commutes with all of $\mathcal{D}_n$}\\
        &= \left\{ d_2^{-1}(l) : l \in [\ncol{n}] \text{ and } c\alpha(l) \neq \alpha c (l) \right\} \tag*{\hspace{-1cm}writing $l := d_2(i)$} \\
        &= d_2^{-1}\cdot \mathrm{Cuts}(\alpha) .
    \end{align*}
    The second result is a simple application of the first, setting $d_2$ equal to the identity.
\end{proof}

Using the cuts set, we can now conveniently express sets of rearrangements in terms of the number of cuts applied by each rearrangement.

\begin{center}\begin{minipage}{\linewidth}
\begin{lemma} \label{lem:sets_of_rearrangements}
The cuts set describes the following classes of rearrangements.
\begin{enumerate}[(i) ]
    \setlength\itemsep{-1em}
    \item Inversions of some number of adjacent regions:
        \vspace{-3mm} \[\mathrm{SignedInversions}(n) := \{ \textbf{z}\alpha \in \mathbb{C}[\HO_n] : |\mathrm{cuts}(\alpha)| = 2\}\]
    \item Transpositions, including where one or more of the segments are inverted:
        \vspace{-3mm} \[\mathrm{SignedTranspositions}(n) := \{ \textbf{z}\alpha \in \mathbb{C}[\HO_n] : |\mathrm{cuts}(\alpha)| = 3 \}\]
    \item Block interchanges (two segments are swapped and/or inverted):
        \vspace{-3mm} \[\mathrm{SignedBlockInterchanges}(n) := \{ \textbf{z}\alpha \in \mathbb{C}[\HO_n] : |\mathrm{cuts}(\alpha)| \in \{3,4\} \}\]
\end{enumerate}
\hfill\qed
\end{lemma}
\end{minipage}\end{center}

\subsection{How many rearrangements \texorpdfstring{$\textbf{z}\sigma$}{(cosets)} perform \textit{k} cuts?}

\begin{table}
\begin{center}
    \setlength{\tabcolsep}{2.5pt}
    \begin{tabular}{ c|cc|cccccccc } 
      & & & \multicolumn{8}{|c}{\# $\mathcal{D}_n\alpha$ performing $k$ cuts} \\
      n & $|\HO_n|$ & \# $\mathcal{D}_n\alpha$ & 0 & 1 & 2 & 3 & 4 & 5 & 6 & 7\\
      \hline
      2 & 8 & 2 & (1, 1) & (0, 0) & (1, 1) &  &  &  &  & \\ 
      3 & 48 & 8 & (1, 1)& (0, 0) & (3, 3) & (1, 4) &  &  &  & \\
      4 & 384 & 48 & (1, 1) & (0, 0)& (6, 6) & (4, 16) & (1, 25) &   &   &  \\
      5 & 3840 & 384 & (1, 1) & (0, 0)& (10, 10) & (10, 40) & (5, 125) & (1, 208) &   &  \\
      6 & 46080 & 3840 & (1, 1)& (0, 0) & (15, 15) & (20, 80) & (15, 375) & (6, 1248) & (1, 2121) &  \\
      7 & 645120 & 46080 & (1, 1)& (0, 0) & (21, 21) & (35, 140) & (35, 875) & (21, 4368) & (7, 14847) & (1, 25828) \\
     \hline
    \end{tabular}
    \caption{The number of distinct rearrangements $\mathcal{D}_n\alpha$ with a given number of cuts. For each tuple $(p,q)$ in the table (corresponding to a fixed $n$ and $k$), $p = \binom{n}{k}$ represents the number of distinct cut-sets among the rearrangements and $q$ the total number of rearrangements.}
    \label{tbl:cuts_single_cosets}
\end{center}
\end{table}

\begin{table}
\begin{center}
    \setlength{\tabcolsep}{4pt}
    \begin{tabular}{ c|cc|cccccccc } 
      & & & \multicolumn{8}{|c}{\# $\mathcal{D}_n\alpha$ per cut set of size $k$} \\
      n & $|\HO_n|$ & \# $\mathcal{D}_n\alpha$ & 0 & 1 & 2 & 3 & 4 & 5 & 6 & 7\\
      \hline
      2 & 8 & 2 & 1 &0 & 1 & &  &  & & \\
      3 & 48 & 8 & 1 &0 & 1 & 4  & &  & & \\
      4 & 384 & 48 & 1 &0 & 1 & 4 & 25  &  & & \\
      5 & 3840 & 384 & 1 &0 & 1 & 4 & 25 & 208  & & \\
      6 & 46080 & 3840 & 1 &0 & 1 & 4 & 25 & 208 & 2121 & \\
      7 & 645120 & 46080 & 1 &0 & 1 & 4 & 25  & 208 & 2121 & 25828 \\
     \hline
    \end{tabular}
    \caption{Number of distinct rearrangements for each distinct cut-set of size $k$ (q/p in Table 3)}
    \label{tbl:cuts}
\end{center}
\end{table}

The cuts set can be a useful way of thinking about rearrangements. In \Cref{tbl:cuts_single_cosets}, we count the number of rearrangements which perform a given set of $k$ cuts (that is, for each $k$ we count the number of cut-sets $S$ of size $k$ and then count the number of rearrangements $\D_n \alpha$ such that $\cuts(\D_n \alpha) = S$). We note that for each $(p,q)$ entry in \Cref{tbl:cuts_single_cosets}, we must have $\frac{q}{p}$ distinct rearrangements for each particular cut-set, since the number of ways to reconnect a genome cut into $k$ segments is independent of the positions of the cuts and the region labels (that is, the number of rearrangements with a given cut-set is dependent only on the size of that cut-set). In \Cref{tbl:cuts}, we collect $\frac{q}{p}$ for each entry in \Cref{tbl:cuts_single_cosets}.

We can observe a number of patterns emerging in \Cref{tbl:cuts_single_cosets}. First, as one would expect, it seems that the number of distinct cut-sets for rearrangements on $n$ regions is given by $\binom{n}{k}$ where $k$ is the number of cuts. The total number of rearrangements which perform $k$ cuts is less obvious, however we can see from \Cref{tbl:cuts_single_cosets} that the number of rearrangements for a given set of $k$ cuts ($\frac{q}{p}$ in the table) is independent of $n$, and begins as \[
  1,0,1,\frac{140}{35}=4,\frac{875}{35}=25,\frac{4368}{21}=208,\frac{14847}{7}=2121,\frac{25828}{1} = 25828,...\\
\]
coinciding with the sequence A$061714$ in the \textit{Online Encyclopedia of Integer Sequences} \cite{OEISFoundation2020}. For small $n$, we can think about how we might count the number of rearrangements which correspond to a given cut set. Consider $k=3$, and consider some fixed set of cuts. Such a rearrangement will split a genome into three parts. The rearrangement will re-connect these segments in one of four different ways:
\begin{figure}[H]
    \centering
    \begin{center}\includegraphics[width=310px]{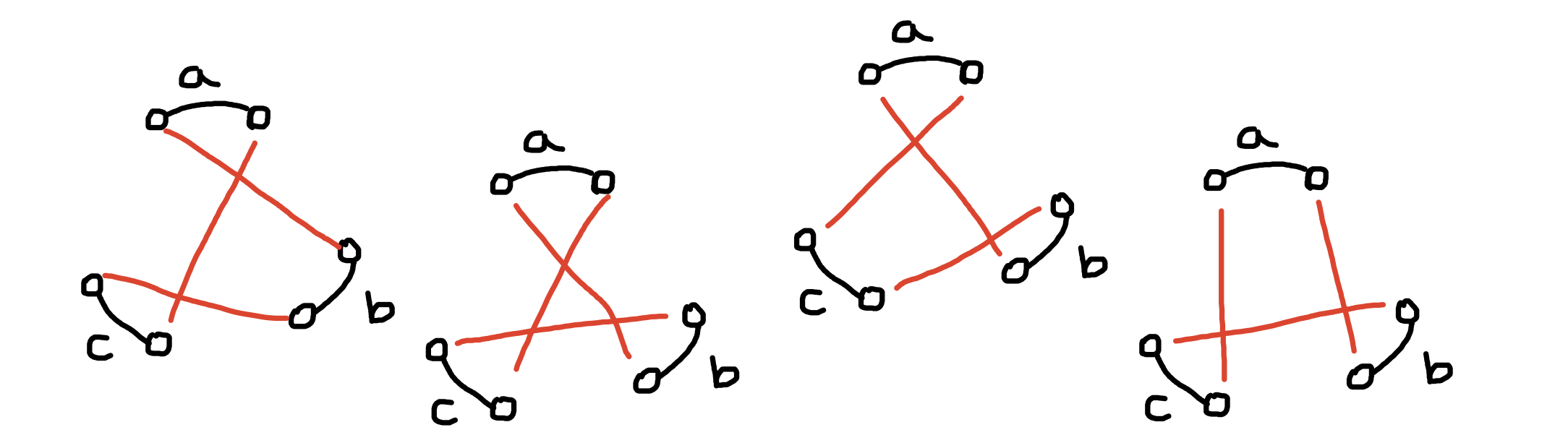}\end{center}
    \vspace{-4mm}
    \caption{Different ways to reconnect three labelled segments to perform a rearrangement which requires 3 cuts.}
    \label{fig:counted}
\end{figure}
\noindent Of course, there are other ways to reconnect these segments, for example the graphs in \Cref{fig:not_counted}.
\begin{figure}[H]
    \centering
    \begin{center}\includegraphics[width=240px]{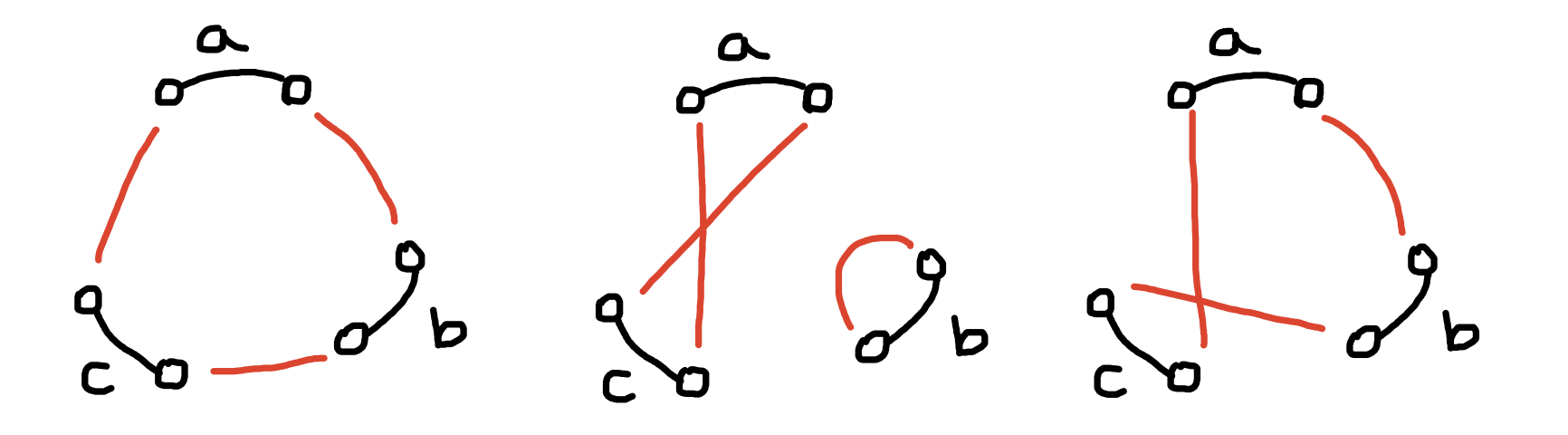}\end{center} \vspace{-4mm}
    \caption{From left to right: diagrams representing the identity rearrangement, a fission, and a rearrangement performing 2 cuts.}
    \label{fig:not_counted}
\end{figure}
\noindent However, these aren't counted for various reasons. A rearrangement matching the first picture is simply the identity rearrangement. Similarly, the last picture represents a rearrangement which can be performed with just two cuts, rather than three. Finally, the middle picture represents what is referred to as a fission, which for now we do not consider. If one includes examples such as the left and right images above, it is easy to see that we obtain 8 possible rearrangements ($8=2^{3-1}(3-1)!$) which is of course the total number of genomes with $3$ regions. The sequence A061714 however is more complicated, with the $k$-th term given by:
\[\mathrm{A061714}(k) = (-1)^{k}+\sum_{i=0}^{k-1}(-1)^{k-1+i}\left(\begin{array}{c}k \\ i+1\end{array}\right) i!\  2^{i}, \quad \text{ for } k \geq 2.\]
To see that this sequence gives the number of rearrangements $\textbf{z}\alpha$ with a given $k$-cut set in general, note that \citeauthor{Helsgaun2009} \cite{Helsgaun2009} identifies the sequence $\mathrm{A061714}$ with the number of \textit{pure $k$-opt moves} which, in the context of the travelling salesman problem implements the following: remove $k$ edges from a tour, and replace them with $k$ different ones. In fact, one can observe that the example figures describing such moves in \cite{Helsgaun2009} are quite similar to our figures above. We can of course now write the number of rearrangements which perform $k$ cuts.
\begin{proposition}
    The number of rearrangements $\D_n \alpha \in \HO_n/\D_n$ such that $|\cuts(\alpha)| = k$ is
    \[{\binom{n}{k}}\mathrm{A061714}(k) \ =\  {\binom{n}{k}} (-1)^{k}+ {\binom{n}{k}}\sum_{i=0}^{k-1}(-1)^{k-1+i}\left(\begin{array}{c}k \\ i+1\end{array}\right) i!\  2^{i}.\vspace{-6mm}\]
    \hfill\qed
\end{proposition}

The interested reader can verify via the generating function for $A061714$ that the sum of the above expression over all $k$ is equal to the number of cosets $\HO_n/\D_n$.

\subsection{How many distinct \textit{k}-cut rearrangement actions are there?}
As described in \Cref{sec:ER}, it makes sense to think about rearrangement actions as algebra elements corresponding to double cosets. \Cref{prop:cuts_set} shows that each double coset admits a set of cut-sets which are all of the same size. We can therefore identify distinct rearrangement actions which perform a given number of cuts. The number of actions performing $k$ cuts for $n$ regions, $n \leq 7$ is summarised in \Cref{tbl:cuts_double_cosets}. 

While we have no closed form expression for the number of such actions, we can now write the number of distinct rearrangement actions performing $2$ or $3$ cuts. These respectively correspond to inversions and transpositions (\Cref{lem:sets_of_rearrangements}) which seem to be the most biologically reasonable and the most well-studied in our specific context. First, we have the following proposition,

\begin{proposition}\label{prop:num_inversions}
     For a circular genome with $n$ oriented regions, the number of distinct rearrangement actions which apply 2 cuts (that is, the number of algebra elements $\textbf{z}\alpha\textbf{z}$ such that $|\cuts(\alpha)| = 2$) is given by $\floor*{\frac{n}{2}}$. 
\end{proposition}
\begin{proof} Consider the possibilities for making two cuts. There are $n$ possible positions to place a cut, and we want to choose two (this is why there are $\binom{n}{2}$ cosets that achieve this). In the double coset case, some choices are equivalent due to dihedral symmetry of the regions, so we think about the distance between the two cuts (or equivalently, the sizes of the two segments). If $n$ is even, we can choose two cuts that are at most $\frac{n}{2}$ regions apart. If $n$ is odd, the cuts can be at most  $\frac{n-1}{2}$ regions apart,
and so the total number of options in each cases is $\floor{\frac{n}{2}}$. Finally, for each choice of cuts there is only one non-trivial way to put the two segments back together, giving us $\floor{\frac{n}{2}}$ total rearrangements performing $2$ cuts.
\end{proof}

We can similarly consider the number of distinct rearrangement actions that perform precisely 3 cuts. First, choose a configuration of the three cuts along the circular genome. That is, choose a partition of the number of regions $n$ into 3 parts. Each partition leads to a different number of options for reconnecting the three segments. For example, if $n=6$, the partition $n\!=\!2\!+\!2\!+\!2$ (all parts are the same size) leads to two distinct events: inverting every segment, or inverting any two segments. The partition $n\!=\!1\!+\!1\!+\!4$ (two parts the same size) leads to three distinct events: inverting every segment, inverting two segments of size 1, or inverting a segment of size 1 and the segment of size 4. In general, we have the following.

\begin{lemma} \label{prop:put-back-together}
    After breaking a genome into three unlabelled segments, there are 2 distinct actions placing the genome back together again if the parts are all the same size, 3 distinct actions if two segments are the same size, and 4 distinct actions if the segments all have different sizes.\hfill\qed
\end{lemma}

\Cref{prop:put-back-together} can be obtained by considering the possible options for reconnecting parts of a genome, as shown in \Cref{fig:counted}. Further, we can count the number of partitions of $n$ with 0, 2 or 3 parts equal with the following lemma.

\begin{lemma} \label{prop:parts-the-same}
    There are $\floor*{\frac{n-1}{2}}$ partitions of $n$ into $3$ parts which have exactly 2 parts of equal size, and one partition into three equal parts if and only if $3|n$. In other words, there are $1 - \ceil*{\frac{n}{3}} + \floor*{\frac{n}{3}}$ partitions into 3 equal parts. \hfill\qed
\end{lemma}

The number of integer partitions of $n\!-\!3$ into three parts is given by the sequence $\mathrm{A001399}$ \cite{OEISFoundation2020}, and we can use this sequence along with the previous two propositions to obtain the following closed-form expression.

\begin{proposition} \label{prop:3-cut-actions}
    For a circular genome with $n$ oriented regions, the number of distinct rearrangement actions which apply 3 cuts (that is, the number of algebra elements $\textbf{z}\alpha\textbf{z}$ such that $|\cuts(\alpha)| = 3$) is given by,
    \[\floor*{\frac{n^2}{3}}-\floor*{\frac{n\vphantom{n^2}}{2}}.\]
\end{proposition}
\begin{proof}
    Combining \Cref{prop:put-back-together} and \Cref{prop:parts-the-same} we obtain
    \begin{align*}
        4\paren[\Big]{\mathrm{A001399}(n\!-\!3)} - \floor*{\frac{n-1}{2}} + \ceil*{\frac{n}{3}} - \floor*{\frac{n}{3}} - 1&\\ 
        = 4\floor*{\frac{n^2 + 4}{12}} - \floor*{\frac{n-1}{2}} + \ceil*{\frac{n}{3}} - \floor*{\frac{n}{3}} - 1&.
    \end{align*}
    Further, the above can be simplified to obtain the required form via manipulating generating functions.
\end{proof}

We currently have no expression for the number of rearrangement actions requiring four cuts, however we note that inversions and transpositions---which require just 2 or 3 cuts---are the most commonly considered rearrangement actions in the literature, and may be considered the most `biologically reasonable' rearrangements to consider as single events (see \Cref{tbl:models_in_literature}). Adding rearrangements which perform $>3$ cuts into a model, for example some block-interchanges \cite{Huang2010}, also requires additional considerations. For example, these rearrangements will be products of other elements in the model, potentially making assigning probabilities more difficult. Inverses must also be considered, since for inversions $\zed \alpha \zed$ performing more than 3 cuts, the inverse $\alpha^{-1}$ might not be a term in $\zed \alpha \zed$. In this case, $\zed \alpha^{-1} \zed$ must be included in the model if one wishes to ensure that rearrangements can be `undone' in a single step.

\begin{table}
\begin{center}
%\color{blue}
    \setlength{\tabcolsep}{4pt}
    \begin{tabular}{ c|cc|ccccccc } 
      & & & \multicolumn{7}{|c}{\# $\mathcal{D}_n\alpha\mathcal{D}_n$ performing $k$ cuts} \\
      n & $|\HO_n|$ & \#$\mathcal{D}_n\alpha\mathcal{D}_n$ & 0 & 2 & 3 & 4 & 5 & 6 & 7\\
      \hline
      2 & 8 &       2    & 1 & 1 &  &   &  & & \\
      3 & 48 &      4    & 1 & 1 & 2  &   &  & & \\
      4 & 384 &     13   & 1 & 2 & 3  & 7   &  & & \\
      5 & 3840 &    56   & 1 & 2 & 6  & 15 & 32  & & \\
      6 & 46080 &   381  & 1 & 3 & 9  & 41 & 114 & 213 & \\
      7 & 645120 &  3486 & 1 & 3 & 13 & 70 & 342 & 1083 & 1974  \\
     \hline
    \end{tabular}
    \caption{The number of double cosets (i.e. rearrangements) on $n$ regions that perform $k$ cuts, for each $n$ and $k$.}
    \label{tbl:cuts_double_cosets}
\end{center}
\end{table}

\subsection{Breakpoints} \label{sec:breakpoints}

Breakpoints are the locations on the genome at which rearrangements can or have occurred. This broad definition is similar to the set of cuts which we defined in \Cref{def:cuts}. The main difference is that the set of cuts keeps track of the \textit{positions} of the regions. While the specific position labels themselves are less important in the circular case (since positions can be rotated/flipped around without altering the genome) the cut set does describe a configuration of cuts around the genome. Breakpoints have a number of different definitions in the literature, some of which are specific to particular symmetries or other assumptions \cite{Sankoff2004, Fertin2009, Alexandrino2021}. In this section we provide a definition of a breakpoint for circular genomes that is analogous to the set of cuts.

\begin{definition} \label{def:breakpoint}
    Let $\sigma$ be an \textbf{oriented} genome with $n$ regions. The set of breakpoints is defined as,
    \[\mathrm{breakpoints}(\sigma) = \set{ (i, j) \in [n]^2 : \text{regions $i$ and $j$ are out of order with respect to the reference}  }.\vspace{-0.6cm}\]
\end{definition}

We can rewrite this definition more algebraically, as we did for the cuts set in \Cref{def:cuts}.

\begin{proposition}\label{def:breakpoint}
    Let $\sigma$ be an \textbf{oriented} genome with $n$ regions. The set of breakpoints can be written as,
    \begin{align*}
        \mathrm{breakpoints}(\sigma) &= \left\{ (\sigma^{-1}(i),\sigma^{-1} c(i)) : i\in [n] \text{ and } \sigma^{-1} c(i) \neq c\sigma^{-1}(i) \right\}\\
        &= \left\{ (\sigma^{-1}(i),\sigma^{-1} c(i)) : i\in [n] - \mathrm{Fix}([c,\sigma^{-1}])\right\}.
    \end{align*}
\end{proposition}

We can see from \Cref{def:breakpoint} that the number of breakpoints for a genome $\sigma$ can be written
\[ 
    \mathrm{bp}(\sigma) := \left|\mathrm{breakpoints}(\sigma)\right| = \frac{2n-|\mathrm{Fix}([\sigma^{-1},c])|}{2}.
\]
This is equivalent to a statement in a paper by \citeauthor{Meidanis2000} \cite{Meidanis2000}, in which genomes are written in the content paradigm. For a genome $\tau$ written in the content paradigm, we have
\[ \mathrm{bp}(\sigma) = \frac{2n-|\mathrm{Fix}(r\tau r c)|}{2}.\]
To show these are equivalent, note that a genome $\tau$ in the content paradigm can be written
\[ \tau = \sigma^{-1} c \sigma, \]
where $\sigma$ represents the same genome represented by $\tau$, but in the position paradigm, as a map from regions to signed positions. We have,
\begin{align*}
    r \tau r c &= r \sigma^{-1} c \sigma r c  \\
    &=  r (\sigma^{-1} c^{-1} \sigma)^{-1}  r^{-1} c \tag*{(\text{since $r^{-1} = r$)}} \\
    &= ( r (\sigma^{-1} c^{-1} \sigma)  r^{-1})^{-1} c \\
    &= \sigma^{-1} c^{-1} \sigma c \tag*{(\text{one can show \(r (\sigma c^{-1} \sigma^{-1})  r^{-1} = (\sigma c^{-1} \sigma^{-1})^{-1}\)})} \\
    &= [\sigma, c],
\end{align*}

and it is easy to see that 
\begin{align*}
    |\mathrm{Fix}([\sigma, c])| = |\mathrm{Fix}([\sigma^{-1}, c])|.
\end{align*}
Thus the two statements are equivalent. 

\section{Discussion}
While model-based approaches to computing genome rearrangement can lead to to significant computational challenges, similar approaches have been proven to be fruitful in other areas of phylogenetics. In particular, we contend that the development of these techniques will inevitably open up more possibilities for new distance estimates, lead to a deeper understanding of genome rearrangements, and allow us to compare and evaluate existing distances in a broader context. 

The genome rearrangements literature covers a broad range of combinatorial approaches which have proven to be effective for a range of definitions of the genome rearrangement distance. However, these methods don't yet benefit from the higher-level understanding that can be obtained from an algebraic perspective. Conversely, existing algebraic, model-based approaches could be improved by incorporating and unifying some of the knowledge contained within the combinatorial literature. We have reviewed existing literature, and provided examples and formulations which begin to strengthen this connection. With model-based approaches come the problems of identifying and defining sensible models, including possible rearrangements and associated probabilities. The ideas and examples we have presented here help to simplify this process.

There are of course a number of challenges which need to be overcome in order to further the development of algebraic, model-based approaches for computing rearrangement distances. The primary obstacle is the inherent computational complexity of these methods, due to the factorially large number of genomes. For example, as a distance, the maximum likelihood estimate of time taken to transform one genome into another incorporates all possible paths between genomes and can be computed under any rearrangement model without modification, however is extremely computationally intensive. The MLE distances can be made more tractable using representation theory \cite{Terauds2018}, and other related measures have similar benefits but are easier to compute, for example the mean first passage time (MFPT) \cite{Francis2014, Terauds2021}, which can also be computed under any rearrangement model. Such methods have so far not had the same level of attention in the literature as combinatorial approaches which focus on a single model of genome rearrangement, and there exists a wealth of knowledge of Markov chain theory and related concepts which may be applied in this area.
We hope that formalising existing concepts from the literature will help to foster the development of new ideas, eventually allowing us to overcome some of the computational challenges inherent in genome rearrangement modelling problems.

\section*{Data availability statement}
Data sharing is not applicable to this article as no datasets were generated or analysed during the current study.

\addcontentsline{toc}{section}{References}
\scriptsize
\bibliographystyle{plainnat}
\bibliography{genomes_refs}

\end{document}